\documentclass[11pt]{article}

\usepackage[a4paper,margin=1in]{geometry}
\usepackage[T1]{fontenc}
\usepackage[utf8]{inputenc}
\usepackage{lmodern}
\usepackage{natbib}
\usepackage{graphicx}
\usepackage{amsmath,amssymb,amsfonts,amsthm}
\usepackage{mathrsfs}
\usepackage{booktabs}
\usepackage{multirow}
\usepackage{array}
\usepackage{caption}
\usepackage{enumitem}
\usepackage{xcolor}
\usepackage{url}
\usepackage[colorlinks=true,linkcolor=blue,citecolor=blue,urlcolor=blue]{hyperref}

\newtheorem{theorem}{Theorem}
\newtheorem{proposition}[theorem]{Proposition}

\newcommand{\sign}{\operatorname{sign}}
\newcommand{\HS}{\mathrm{HS}}

\title{\textbf{When Similarity Is Interaction-Driven: Quantum Kernels for Regime-Sensitive Learning}}

\author{
Hanqiu Peng
\and
Jianlong Lu
\and
Ying Chen\thanks{Corresponding author: matcheny@nus.edu.sg}\\[0.3em]
Centre for Quantitative Finance, Department of Mathematics \\
Risk Management Institute, National University of Singapore, Singapore
}

\date{}

\begin{document}

\maketitle

\begin{abstract}
Similarity in many decision systems is governed not by distance alone but by interactions among variables. In fraud and anomaly detection, small local perturbations can cross interaction-sensitive decision boundaries while leaving ambient distance almost unchanged. Motivated by this setting, we introduce a thin-slab interaction model and an interaction-driven quantum kernel constructed from entangled Pauli-string feature maps. The feature map explicitly encodes sparse high-order block interactions. We show that the resulting fidelity kernel is positive semidefinite, admits an exact block-factorized formulation, and induces a geometry sensitive to changes in interaction regime. Across balanced and imbalanced synthetic experiments spanning third-, fourth-, sixth-, and eighth-order interactions, the proposed kernel consistently outperforms linear, radial basis function, Laplacian, and polynomial kernels, as well as an engineered-interaction linear baseline supplied with the planted block products. On real fraud-detection benchmarks, it achieves the highest mean accuracy and F1 on Credit Card Fraud Detection and ranks second on IEEE-CIS Fraud Detection. Executed on a 156-qubit IBM Quantum processor in a fourth-order setting, the hardware-estimated kernel matches the noise-free simulator within seed-to-seed variability and retains its advantage over the baselines. These findings show that quantum-kernel performance depends on alignment between feature-map geometry and the underlying predictive structure, rather than on Hilbert-space dimension alone. Because the prescribed block-factorized kernel can also be evaluated exactly on a classical computer, the results establish predictive and representational value rather than computational quantum speedup.
\end{abstract}


\section{Introduction}
\label{sec:introduction-qk}
In many fraud and anomaly-detection problems, predictive similarity is governed by combinations of contextual variables rather than by distance in the original feature space alone. Two transactions can be close in amount, location, and timing yet belong to different risk regimes because of a small set of interacting conditions. This form of contextual anomalousness is consistent with the broader anomaly-detection literature \citep{chandola2009anomaly,hilal2022financial}. Fraud detection also involves class imbalance, distribution shift, evolving adversarial behavior, and heterogeneous transaction contexts \citep{dalpozzolo2015credit,lucas2020credit}. These features motivate similarity measures that respond to sparse high-order interactions rather than only to coordinate-wise proximity.

Quantum kernels provide a circuit-defined approach to constructing similarity. Rather than measuring proximity only in the original feature space, quantum feature maps can encode structured high-order interactions directly into entangled quantum states, inducing a geometry naturally sensitive to interaction-driven behavior. Motivated by fraud and anomaly detection, this study examines whether interaction-aligned quantum kernels can provide a more informative representation of hidden behavioral regimes than conventional distance-based kernels. Our central hypothesis is simple: \textbf{when predictive structure is governed by interaction-sensitive decision boundaries rather than distance alone, kernels aligned with those interactions can substantially improve the representation of risk and regime similarity.}

Kernel methods derive their predictive power from the geometry induced by their similarity functions. Linear kernels encode coordinate-wise relationships, while radial basis function (RBF) and Laplacian kernels define similarity through input-space distance. Polynomial kernels can represent high-order interactions, but generic polynomial kernels do not specifically privilege the sparse block products relevant here. These inductive biases are effective in many learning problems, but can become poorly aligned when predictive regimes are separated by sparse interaction boundaries rather than large geometric distances. In such settings, nearby observations may belong to entirely different behavioral regimes, whilst distant observations may share the same underlying interaction structure. The central challenge is therefore not simply nonlinear representation, but the construction of a similarity geometry aligned with the true predictive mechanism.

Recent work in both classical and quantum machine learning has increasingly highlighted the importance of inductive-bias alignment \citep{havlicek2019supervised,schuld2019quantum,kubler2021inductive,huang2021power,caro2022generalization}. High-order interactions are particularly difficult because the relevant predictive signal is often sparse and structurally localized. This challenge is well recognized in classical interaction learning, where detecting and selecting interaction terms becomes increasingly difficult as the number of candidate feature combinations grows \citep{sorokina2008detecting,bien2013lasso}. Polynomial kernels can theoretically represent such interactions, but the relevant terms are embedded within combinatorially large and largely unstructured feature expansions. Under finite samples, learning may fail to exploit the relevant terms effectively because they are embedded amongst many irrelevant monomials, making sparse interaction regimes difficult to isolate through generic kernel constructions. Related models such as factorization machines and neural factorization machines address this issue by explicitly parameterizing or learning feature interactions in sparse predictive settings \citep{rendle2010factorization,he2017neural}. This motivates the search for feature maps whose induced geometry is explicitly aligned with interaction-driven predictive structure.

To investigate this question systematically, we introduce a thin-slab interaction framework that separates interaction-driven predictive structure from generic distance-based similarity. Within this framework, an entangled Pauli-string feature map encodes fixed block-wise high-order products. The resulting fidelity kernel is positive semidefinite, admits an exact block-factorized formulation, and depends explicitly on coordinate-level phase differences and block-product discrepancies. Because the block partition is fixed and known, the implementation evaluates the block states separately and avoids constructing a statevector spanning all input features. This is a block-factorized quantum-circuit representation of the prescribed interaction geometry, although the same known-block structure also admits efficient classical evaluation. The objective is therefore feature-map alignment rather than representation dimension alone, consistent with analyses that relate quantum-kernel performance to the target function class and kernel spectrum \citep{kubler2021inductive,huang2021power,caro2022generalization}.

Our work is related to the growing literature on quantum kernel methods and quantum-enhanced feature spaces \citep{havlicek2019supervised,schuld2019quantum}. Prior studies have investigated generalization properties \citep{caro2022generalization}, effective dimension \citep{abbas2021power}, concentration phenomena \citep{thanasilp2024exponential}, and the importance of data structure in quantum machine learning \citep{huang2021power,kubler2021inductive}. Rather than pursuing a universal quantum-advantage claim, our work focuses on a more specific and practically relevant question: under what predictive structures can quantum-induced similarity geometry provide predictive benefits over conventional kernels? In particular, we argue that interaction-driven decision boundaries provide a natural setting in which quantum feature maps may offer substantial representational benefits.

Empirically, we adopt a three-stage evaluation design. First, controlled synthetic experiments test the interaction mechanism at third and fourth order. Second, a higher-order scaling study tests whether its predictive effectiveness persists as input dimension and interaction order increase. Third, two fraud-detection benchmarks assess applicability beyond the controlled synthetic setting using data from the Fraud Dataset Benchmark (FDB) \citep{grover2022fraud}.

Across the synthetic settings, from the controlled third- and fourth-order tasks to sixth- and eighth-order scaling, the proposed quantum kernel consistently outperforms standard linear, RBF, Laplacian, and polynomial kernels. It also outperforms \texttt{InterFea}, a ridge-regularized linear classifier trained on the original features augmented with the planted block products, in every evaluated configuration---by accuracy in balanced tasks and F1 in imbalanced tasks. This shows that the fidelity geometry adds value beyond direct linear access to those products. In a more interaction-sensitive balanced setting, the quantum kernel raises accuracy from approximately $0.50$--$0.52$ for standard kernels to over $0.75$, a relative improvement of more than $45\%$. In imbalanced settings, it raises F1 from approximately $0.20$--$0.28$ to above $0.53$, more than doubling performance in several regimes.

In the scaling study, the quantum kernel achieves the highest mean accuracy and F1 in both the 60-dimensional sixth-order and 80-dimensional eighth-order settings, with $(\mathrm{Acc.},\mathrm{F1})=(0.706,0.412)$ and $(0.695,0.395)$, respectively. Relative to \texttt{InterFea}, the paired mean F1 gaps are $0.032$ and $0.075$. Together with the controlled third- and fourth-order results, these findings establish a persistent advantage across all evaluated orders, with the largest noise-free imbalanced F1 gap at eighth order.

On the real fraud-detection datasets, the method achieves the strongest F1 on Credit Card Fraud Detection, improving mean F1 from $0.348$ for the best competing classical kernel to $0.472$, corresponding to a relative improvement of approximately $36\%$. It also remains competitive on IEEE-CIS Fraud Detection, outperforming most conventional kernels and ranking second to the Laplacian kernel. Together, the controlled synthetic results demonstrate a quantum-induced predictive advantage across progressively more demanding interaction regimes, while the real-data results show that the method remains competitive beyond the controlled synthetic setting.

Finally, we execute the proposed kernel on a superconducting quantum processor. On the 156-qubit IBM Quantum processor \texttt{ibm\_kingston}, in the $(d,k,B)=(20,4,5)$ setting with $\epsilon=0.15$, the hardware-estimated kernel attains mean balanced accuracies of $0.749$ and $0.745$ and mean imbalanced F1 scores of $0.531$ and $0.510$ at $p_{\mathrm{flip}}=0$ and $0.10$, respectively. These values lie within $0.007$ of the noise-free simulator and above \texttt{InterFea} in all four configurations.

The main contributions of this paper are as follows:
\begin{itemize}
    \item We introduce a thin-slab interaction model as a controlled abstraction of interaction-sensitive decision boundaries, where Euclidean proximity becomes misaligned with predictive similarity.
    \item We propose an interaction-driven quantum kernel based on entangled Pauli-string feature maps that explicitly encode sparse high-order interactions.
    \item We show theoretically that the resulting fidelity kernel is positive semidefinite, admits an exact block-factorized formulation, and induces a geometry naturally sensitive to interaction-driven behavior.
    \item We quantify the resources required by the block-factorized implementation, including reusable qubits, logical gate counts, circuit depth, statevector storage, and kernel-matrix memory.
    \item We evaluate the kernel on a 156-qubit IBM Quantum processor, where its predictive performance matches the noise-free simulator within seed-to-seed variability.
    \item We provide systematic numerical evidence across multiple interaction orders, class-balance regimes, and scaling settings, showing consistent predictive gains over conventional distance-based kernels and an explicit known-block interaction baseline.
    \item We evaluate the method on two real fraud-detection benchmarks, where it achieves strong or competitive performance beyond the controlled synthetic setting.
    \item More broadly, we show that the effectiveness of quantum kernels depends critically on inductive-bias alignment between feature-map geometry and the underlying predictive structure, highlighting the importance of interaction-driven representations in quantum machine learning.
\end{itemize}

\paragraph{Scope of the contribution.}
This paper builds on quantum kernels, inductive-bias analysis, and classical high-order interaction learning. Existing work has established that quantum feature maps can induce rich similarity geometries and that their practical value depends on feature-map design rather than Hilbert-space dimension alone. The contribution here is a controlled setting in which the required alignment can be stated and tested directly: the thin-slab construction makes Euclidean proximity a poor proxy for label similarity, and the proposed feature map is designed to match that structure rather than to maximize expressivity. The paper therefore demonstrates a quantum-induced predictive advantage in the evaluated regimes, not universal or computational quantum advantage. Because the disjoint-block kernel studied here is exactly classically evaluable, the advantage established is representational and statistical rather than computational.

The remainder of this paper is structured as follows. Section~\ref{sec:relatedwork-qk} reviews related work on quantum kernels, inductive bias, and high-order interaction learning. Section~\ref{sec:problemsetup-qk} introduces the thin-slab interaction model, and Section~\ref{sec:qkernel-qk} develops the proposed interaction-driven quantum kernel and its theoretical properties. Section~\ref{sec:experimentaldesign-qk} describes the experimental protocol, and Section~\ref{sec:results-qk} reports empirical evaluations on synthetic interaction-learning tasks, quantum hardware, and real fraud-detection benchmarks. Section~\ref{sec:discussion-qk} discusses implications, limitations, and directions for future research, and Section~\ref{sec:conclusion-qk} concludes.
\section{Related Work}
\label{sec:relatedwork-qk}

\subsection{Quantum kernels and inductive bias}

Quantum kernel methods map classical data into quantum states and define similarity through state overlap \citep{havlicek2019supervised,schuld2019quantum}. Early work focused primarily on the representational richness of quantum feature spaces and the possibility of quantum-enhanced learning through high-dimensional embeddings. Subsequent work has shown, however, that practical performance depends strongly on feature-map design, data structure, and the geometry induced by the resulting kernel, rather than on Hilbert-space dimension alone. Related studies have investigated effective dimension \citep{abbas2021power}, high-dimensional implementation on noisy quantum processors \citep{peters2021machine}, exponential concentration of quantum kernels \citep{thanasilp2024exponential}, and trainable or optimized quantum-kernel constructions \citep{pellowjarman2024hybrid,xu2024trainable}.

Our work follows this feature-map-design perspective. Rather than treating quantum kernels as generic high-dimensional embeddings, we investigate whether quantum feature maps can provide advantages when explicitly aligned with interaction-sensitive predictive structure. In particular, we focus on sparse high-order interaction regimes in which Euclidean proximity becomes a poor proxy for predictive similarity. The proposed Pauli-string feature map is therefore designed not for maximal expressivity, but for structured interaction encoding and interaction-aware similarity geometry. This is consistent with recent work on quantum feature-map optimization and kernel-target alignment, where the goal is to design quantum kernels whose induced similarity geometry better matches the supervised learning task \citep{pellowjarman2024hybrid,xu2024trainable}.

\subsection{High-order interactions and similarity geometry}

High-order interactions are central to many learning problems, including anomaly detection, recommendation systems, genomics, and financial risk modeling. However, such interactions are difficult to identify when embedded implicitly within large nonlinear feature spaces. Linear models fail when predictive structure contains little first-order signal, whilst distance-based kernels smooth according to the geometry of the original features, which may be unrelated to hidden interaction regimes. Polynomial kernels can theoretically represent high-order interactions, but their feature spaces contain all monomials up to a given degree. The number of such monomials grows combinatorially with dimension and interaction order, making the relevant sparse terms difficult to isolate within an unstructured feature space under finite samples.

Classical machine learning has long studied ways to represent and select interaction structure more efficiently. Functional ANOVA and ANOVA kernels provide decompositions of functions or kernels into lower- and higher-order interaction components \citep{wahba1990spline,stitson1999support}. Sparse and hierarchical interaction models aim to select informative interactions while controlling model complexity \citep{bien2013lasso,lim2015learning}. Tree-based interaction-detection methods provide another route for identifying non-additive feature effects \citep{sorokina2008detecting}. In sparse predictive settings, factorization machines and neural factorization machines explicitly parameterize feature interactions without constructing a full unstructured polynomial expansion \citep{rendle2010factorization,he2017neural}.

Our work is conceptually related to this direction, but differs in its use of quantum-induced similarity geometry. Instead of explicitly enumerating interaction features, the proposed quantum kernel encodes structured block-wise interactions directly into entangled feature-map phases. This creates a similarity measure that depends explicitly on hidden interaction regimes rather than Euclidean proximity in the original feature space.

\subsection{Fraud detection and interaction-sensitive learning}

Fraud detection and anomaly monitoring provide natural settings for interaction-sensitive learning because predictive regimes are often determined by subtle behavioral combinations rather than large feature deviations. Fraudulent transactions are frequently designed to remain locally similar to legitimate activity, making Euclidean distance a poor indicator of predictive similarity. Small contextual changes, such as device identity, timing, behavioral history, or transaction sequence, may abruptly alter the underlying risk regime despite minimal geometric change in the raw feature space. These properties make fraud detection particularly challenging under severe class imbalance, heterogeneous behavioral patterns, and evolving adversarial dynamics.

Beyond static transaction-level classification, several studies emphasize behavioral and sequential structure in credit-card fraud detection. For example, sequence-based fraud-detection models formulate the task using transaction histories and recurrent neural networks, showing that temporal context and cardholder behavior can improve detection beyond isolated transaction features \citep{jurgovsky2018sequence}. Other work has developed realistic modeling strategies for credit-card fraud under delayed labels, concept drift, and non-stationary transaction distributions \citep{dalpozzolo2018credit}. These studies support the view that fraud risk is often context-dependent and behavior-sensitive, rather than determined by marginal feature deviations alone.

To support reproducible evaluation across fraud and abuse-detection tasks, \citet{grover2022fraud} introduced the Fraud Dataset Benchmark (FDB), which provides standardized datasets, train-test splits, and evaluation protocols. In this work, we use the Credit Card Fraud Detection and IEEE-CIS Fraud Detection benchmarks within the FDB framework. Importantly, these experiments are not intended to claim that real fraud labels follow the synthetic thin-slab interaction mechanism exactly. Instead, they test whether interaction-aligned quantum kernels remain effective when the underlying predictive structure is not controlled to match the proposed inductive bias.

\section{High-Order Interaction Learning under Thin-Slab Geometry}
\label{sec:problemsetup-qk}

We now define the controlled learning problem used to test whether an interaction-aware quantum kernel can outperform standard classical kernels. The construction is intentionally simple: the label is generated by planted block-wise high-order interactions, whereas the sampling distribution places many observations close to the coordinate hyperplanes where those interaction signs can flip. This creates a mismatch between Euclidean proximity and label similarity.

\subsection{Block-structured high-order interaction signal}
\label{subsec:blockmodel-qk}

Let $d$ denote the input dimension, let $B$ be the number of blocks, and let $k$ be the number of coordinates in each block. These quantities satisfy
\begin{equation}
    d = Bk.
    \label{eq:dk_relation-qk}
\end{equation}
For an input vector $x\in[-1,1]^d$, we use $B$ disjoint blocks of size $k$, indexed by $b\in\{1,\dots,B\}$:
\begin{equation}
    x = \bigl(x^{(1)}, x^{(2)}, \dots, x^{(B)}\bigr), 
    \qquad 
    x^{(b)} = (x_{b,1},\dots,x_{b,k})\in[-1,1]^k .
    \label{eq:block_partition-qk}
\end{equation}
For each block $b\in\{1,\dots,B\}$, define a block-level regime bit by the sign of the $k$-way product
\begin{equation}
    r_b(x) = \sign\!\left(\prod_{j=1}^{k} x_{b,j}\right) \in \{+1,-1\}.
    \label{eq:block_regime_bit-qk}
\end{equation}
Let $w_b>0$ denote the fixed weight of block $b$, and let $\tau$ denote a decision threshold. The binary label is generated by the weighted majority vote
\begin{equation}
    y(x) = \sign\!\left(\sum_{b=1}^{B} w_b r_b(x)-\tau\right).
    \label{eq:global_label-qk}
\end{equation}
In balanced experiments we set $\tau=0$ with equal block weights. In imbalanced experiments we keep equal block weights but use a stricter positive-class threshold, so that $y=+1$ only when the block-regime score is sufficiently large. In the reported imbalanced runs we use $\tau=3$ for both $(d,k,B)=(9,3,3)$ and $(20,4,5)$; thus, in the $B=3$ case the positive class occurs only when all three block-regime bits are positive, while in the $B=5$ case it occurs when at least four of the five block-regime bits are positive. If the argument of the sign function is exactly zero, we assign $y=+1$; this convention has negligible effect under the continuous sampling scheme.

This model is explicitly interaction-driven. The relevant information in block $b$ is not contained in any single coordinate $x_{b,j}$, nor in pairwise distances between samples. Instead, the signal is encoded in the sign of the planted $k$-way monomial
\begin{equation}
    m_b(x)=\prod_{j=1}^{k}x_{b,j}.
    \label{eq:block_monomial-qk}
\end{equation}
For $k>2$, this produces a target structure that is invisible to linear models and not directly represented by polynomial kernels whose degree is below $k$.

\subsection{Thin-slab distribution}
\label{subsec:thinslab-qk}

We define the thin-slab distribution explicitly. Let $\epsilon\in(0,1)$ denote the slab thickness. For each block $b$, select one active coordinate
\begin{equation}
    j_b^\star \sim \operatorname{Unif}\{1,\dots,k\},
    \label{eq:active_coordinate-qk}
\end{equation}
and then draw
\begin{equation}
    x_{b,j_b^\star} \sim \operatorname{Unif}([-1,1]),
    \qquad
    x_{b,j} = \epsilon\xi_{b,j},\quad j\neq j_b^\star,
    \qquad
    \xi_{b,j}\overset{\mathrm{i.i.d.}}{\sim}\operatorname{Unif}([-1,1]).
    \label{eq:thin_slab_sampling-qk}
\end{equation}
Here, $\operatorname{Unif}$ denotes the uniform distribution, and the auxiliary variables $\xi_{b,j}$ are independently and identically distributed. When $\epsilon$ is small, most coordinates in each block are close to zero, so each sample lies near several coordinate hyperplanes. The support is a union of thin slabs rather than a full-dimensional uniform cloud.

This sampling scheme amplifies the difference between input-space distance and interaction geometry. Since most inactive coordinates are of order $\epsilon$, a sign flip in one inactive coordinate can reverse the sign of the block product while changing the Euclidean norm by only $O(\epsilon)$. Thus, the label can change sharply across a very small Euclidean perturbation. This is exactly the regime in which a local distance-based kernel may assign high similarity to samples whose block-regime bits differ.

\begin{proposition}[Small Euclidean perturbations can flip block regimes]
\label{prop:small_distance_flip-qk}
Fix a block $b$ and suppose $x,x'\in[-1,1]^d$ differ only in one coordinate $\ell$ of that block. Assume
\begin{equation}
    x_{b,\ell}=\delta,
    \qquad
    x'_{b,\ell}=-\delta,
    \qquad
    \delta>0,
\end{equation}
and $x_{a,j}=x'_{a,j}$ for all $(a,j)\neq(b,\ell)$. If all other coordinates in block $b$ are nonzero, then
\begin{equation}
    r_b(x')=-r_b(x),
\end{equation}
while
\begin{equation}
    \|x-x'\|_2=2\delta.
\end{equation}
In particular, under the thin-slab distribution in~\eqref{eq:thin_slab_sampling-qk}, such regime flips can occur at distance $O(\epsilon)$.
\end{proposition}

\begin{proof}
The two samples differ only by replacing $x_{b,\ell}$ with $-x_{b,\ell}$. Therefore the block product $\prod_{j=1}^k x_{b,j}$ changes sign while all other coordinates remain unchanged. Since the samples differ in exactly one coordinate by $2\delta$, their Euclidean distance is $2\delta$.
\end{proof}

Proposition~\ref{prop:small_distance_flip-qk} explains why the problem is boundary-sensitive. The boundary of a block regime is the union of coordinate hyperplanes
\begin{equation}
    \bigcup_{j=1}^{k}\{x^{(b)}\in[-1,1]^k:x_{b,j}=0\}.
\end{equation}
The thin-slab distribution concentrates many samples near these hyperplanes, making Euclidean smoothness a poor default inductive bias.

\subsection{Why generic kernels can be misaligned}
\label{subsec:classicalfailure-qk}

For two inputs $x$ and $z$, the linear kernel,
\begin{equation}
    k_{\mathrm{lin}}(x,z)=x^\top z,
\end{equation}
is inadequate because the target depends on block-wise products rather than first-order coordinates. Distance-based kernels such as
\begin{equation}
    k_{\mathrm{RBF}}(x,z)=\exp(-\gamma\|x-z\|_2^2),
    \qquad
    k_{\mathrm{Lap}}(x,z)=\exp(-\gamma\|x-z\|_1),
\end{equation}
are more expressive but still define similarity through input-space distance. Proposition~\ref{prop:small_distance_flip-qk} shows that nearby samples can have different block-regime bits, so these kernels can smooth across interaction boundaries.

Here, the positive parameter $\gamma$ controls the distance scale of the RBF and Laplacian kernels and is tuned separately for each kernel family.

Polynomial kernels provide a more serious baseline because they can represent feature interactions:
\begin{equation}
    k_{\mathrm{poly},p}(x,z)=(\gamma x^\top z+c)^p.
\end{equation}
Here, $p$ is the polynomial degree, $\gamma>0$ is a scale parameter, and $c\geq0$ is an offset; all three quantities are specified or tuned as part of the corresponding baseline.
If $p<k$, the planted degree-$k$ terms are absent. When $c\ne0$ and $p\geq k$, the kernel contains monomials of degrees up to $p$, including the relevant degree-$k$ terms, but embeds them within a much larger expansion. The number of monomials of total degree at most $p$ in $d$ variables is
\begin{equation}
    \binom{d+p}{p}.
    \label{eq:monomial_count-qk}
\end{equation}
When $c=0$, the kernel is homogeneous and contains only monomials of total degree exactly $p$; in that case, the planted terms occur directly only when $p=k$. For the inhomogeneous expansion, only $B=d/k$ of the monomials correspond exactly to the planted block products. For example, when $d=20$ and $p=4$, the polynomial feature space contains $\binom{24}{4}=10{,}626$ monomials up to degree four, while only five are the planted block products. This does not mean that polynomial kernels are theoretically unable to represent the target interaction. Rather, under finite samples the relevant block-product terms are embedded in a large unstructured monomial space, so the sparse signal can be difficult to isolate amongst many irrelevant interaction terms. In our experiments, polynomial kernels of several degrees are evaluated under the same validation protocol as the other baselines, but higher-degree kernels remain sensitive to scaling and conditioning, especially under thin-slab sampling where many coordinates have small magnitude. Thus, polynomial kernels are useful but imperfect baselines for sparse $k$-way interaction learning.

This motivates \texttt{InterFea}, an engineered-interaction linear baseline. We augment the raw feature vector with the planted block products,
\begin{equation}
\Phi_{\mathrm{InterFea}}(x)=
\left(x_1,\dots,x_d,\prod_{j=1}^{k}x_{1,j},\dots,\prod_{j=1}^{k}x_{B,j}\right),
\label{eq:engineered_features-qk}
\end{equation}
which induces the linear kernel
\begin{equation}
    k_{\mathrm{InterFea}}(x,z)
    =
    \Phi_{\mathrm{InterFea}}(x)^\top
    \Phi_{\mathrm{InterFea}}(z).
    \label{eq:interfea_kernel-qk}
\end{equation}
In the experiments, the augmented features are standardized using training-set statistics, and an $L_2$-regularized linear least-squares classifier is fitted directly in this engineered feature space using scikit-learn's \texttt{RidgeClassifier}. Binary labels are encoded as $\{-1,+1\}$, and predictions are determined by the sign of the fitted linear score. This primal implementation avoids explicitly constructing the corresponding Gram matrix while retaining the kernel interpretation in~\eqref{eq:interfea_kernel-qk}. InterFea therefore provides a strong control for determining whether the performance gain comes from encoding the relevant interaction features or from the specifically quantum fidelity geometry.

\section{Quantum Feature Map and Kernel Construction}
\label{sec:qkernel-qk}

\subsection{Design hypothesis}
\label{subsec:designhypothesis-qk}

The proposed feature map is based on the following hypothesis: when the target function depends on sparse high-order block interactions, a kernel that explicitly compares samples through these interactions should provide a better-aligned inductive bias than generic kernels that distribute capacity over many irrelevant directions. In the synthetic model, the relevant interaction variables are the block products $m_b(x)=\prod_{j=1}^k x_{b,j}$. The feature map is therefore designed so that differences in $m_b(x)$ directly affect the quantum phase and, consequently, the fidelity kernel. This design does not imply that the resulting kernel is universally superior. Rather, it predicts an advantage only when the encoded block-product structure is aligned with the data-generating mechanism.

\subsection{Block-wise Pauli-string feature map}

For the mathematical construction, let $d=Bk$ and associate one qubit with each input coordinate. For block $b$, let $Z_{b,j}$ denote the Pauli-$Z$ operator acting on the qubit corresponding to coordinate $x_{b,j}$, let $\alpha_b$ denote the interaction-phase scale for that block, and let $\beta$ denote the shared one-body phase scale. Let $H$ denote the Hadamard gate. We define the $d$-qubit feature state
\begin{equation}
    |\psi(x)\rangle=\bigotimes_{b=1}^{B}|\psi^{(b)}(x^{(b)})\rangle,
    \label{eq:tensor_product_feature_state-qk}
\end{equation}
where the $k$-qubit block state is
\begin{equation}
|\psi^{(b)}(x^{(b)})\rangle
=
\exp\!\left(i\alpha_b\left(\prod_{j=1}^{k}x_{b,j}\right)Z_{b,1}\cdots Z_{b,k}\right)
\left[\prod_{j=1}^{k}\exp(i\beta x_{b,j}Z_{b,j})\right]
H^{\otimes k}|0\rangle^{\otimes k}.
\label{eq:block_feature_state-qk}
\end{equation}
Equivalently, the full feature map may be written as
\begin{equation}
|\psi(x)\rangle=
\left[
\prod_{b=1}^{B}
\exp\!\left(i\alpha_b\left(\prod_{j=1}^{k}x_{b,j}\right)Z_{b,1}\cdots Z_{b,k}\right)
\right]
\left[\prod_{i=1}^{d}\exp(i\beta x_iZ_i)\right]
H^{\otimes d}|0\rangle^{\otimes d}.
\label{eq:full_feature_map-qk}
\end{equation}
All phase gates commute because they are diagonal in the computational basis. The one-body phases preserve coordinate-level information, while the block-wise Pauli-string phase encodes the $k$-way product within each block. In the reported experiments, a common interaction scale $\alpha_b=\alpha$ is shared across all blocks and tuned jointly with $\beta$. Although the full feature state is written mathematically as a $d$-qubit state, its factorization across blocks allows the $B$ block states to be evaluated separately. In simulator-based settings, the same $k$ wires are reused; the $(d,k)=(80,8)$ scaling setting instead evaluates the exactly equivalent closed-form state amplitudes. Numerical-backend details are given in Appendix~\ref{secA1-qk}.

The associated quantum kernel is the fidelity kernel
\begin{equation}
    k_Q(x,x')=|\langle\psi(x)|\psi(x')\rangle|^2.
    \label{eq:q_fidelity_kernel-qk}
\end{equation}

\begin{proposition}[Positive semidefiniteness]
The function $k_Q$ in~\eqref{eq:q_fidelity_kernel-qk} is a positive semidefinite kernel.
\end{proposition}

\begin{proof}
For any samples $x_1,\dots,x_N$ and coefficients $c_1,\dots,c_N$,
\begin{align}
\sum_{a,b=1}^{N}c_a^\ast c_b k_Q(x_a,x_b)
&=\sum_{a,b=1}^{N}c_a^\ast c_b |\langle\psi(x_a)|\psi(x_b)\rangle|^2 \\
&=\left\|\sum_{a=1}^{N}c_a |\psi(x_a)\rangle\!\langle\psi(x_a)|\right\|_{\HS}^2\geq 0.
\end{align}
Here, $\|\cdot\|_{\HS}$ denotes the Hilbert--Schmidt norm.
\end{proof}

\subsection{Closed-form block kernel}

Since the block circuit is diagonal after the Hadamard layer, the block state has expansion
\begin{equation}
|\psi^{(b)}(x^{(b)})\rangle
=2^{-k/2}
\sum_{s\in\{\pm 1\}^{k}}
\exp\!\left[
i\beta\sum_{j=1}^{k}x_{b,j}s_j
+i\alpha_b\left(\prod_{j=1}^{k}x_{b,j}\right)\prod_{j=1}^{k}s_j
\right]
|s\rangle,
\label{eq:block_state_expansion-qk}
\end{equation}
where $s_j\in\{\pm1\}$ denotes the $Z$-eigenvalue of qubit $j$. Therefore,
\begin{align}
\langle\psi^{(b)}(x^{(b)})|\psi^{(b)}(x'^{(b)})\rangle
=2^{-k}
\sum_{s\in\{\pm1\}^{k}}
\exp\!\left[
i\beta\sum_{j=1}^{k}(x'_{b,j}-x_{b,j})s_j
+i\alpha_b\{m_b(x')-m_b(x)\}\prod_{j=1}^{k}s_j
\right].
\label{eq:block_overlap_sum-qk}
\end{align}
Let
\begin{equation}
    a_j=\beta(x'_{b,j}-x_{b,j}),
    \qquad
    c_b=\alpha_b(m_b(x')-m_b(x)).
\end{equation}
Using the parity identity separating the two sectors of $\prod_j s_j$, we obtain
\begin{equation}
\langle\psi^{(b)}(x^{(b)})|\psi^{(b)}(x'^{(b)})\rangle
=
\cos(c_b)\prod_{j=1}^{k}\cos(a_j)
+i^{k+1}\sin(c_b)\prod_{j=1}^{k}\sin(a_j).
\label{eq:block_overlap_closed-qk}
\end{equation}
Hence the block contribution to the fidelity kernel is
\begin{equation}
    k_Q^{(b)}(x^{(b)},x'^{(b)})
    =
    \left|
    \cos(c_b)\prod_{j=1}^{k}\cos(a_j)
    +i^{k+1}\sin(c_b)\prod_{j=1}^{k}\sin(a_j)
    \right|^2.
    \label{eq:block_kernel_closed-qk}
\end{equation}
This expression shows explicitly that the kernel depends not only on coordinate-wise differences, but also on block-product discrepancies $m_b(x')-m_b(x)$.

Because the full state factorizes over blocks,
\begin{equation}
    \langle\psi(x)|\psi(x')\rangle
    =\prod_{b=1}^{B}\langle\psi^{(b)}(x^{(b)})|\psi^{(b)}(x'^{(b)})\rangle,
\end{equation}
and therefore
\begin{equation}
    k_Q(x,x')=\prod_{b=1}^{B}k_Q^{(b)}(x^{(b)},x'^{(b)}).
    \label{eq:block_kernel_factorization-qk}
\end{equation}

\subsection{Representation structure and circuit resources}
\label{subsec:representation-efficiency-qk}

The resource comparison below follows the implementations used in the synthetic and real-data experiments. All reported settings use a fixed partition into $B=d/k$ blocks; they do not enumerate all degree-$k$ subsets. For IEEE-CIS, the 60 selected numerical features are ordered by their training-set correlation scores before forming ten blocks of six. Accordingly, the classical \texttt{InterFea} baseline in Eq.~\ref{eq:engineered_features-qk} appends exactly $B$ block products to the original data. Its output dimension is $d+B$, and its construction cost is $O(Bk)=O(d)$ per observation.

For the simulator-based experiments, the PennyLane implementation also exploits the same block factorization. It creates a \texttt{default.qubit} device with only $k$ wires and prepares the $B$ block states separately, so the peak simulated qubit requirement is $k$, rather than $d$. The $(d,k)=(80,8)$ setting uses the exact closed-form numerical backend described in Appendix~\ref{secA1-qk}; the resource counts below nevertheless describe the quantum circuit represented by the same feature map. For one block, that circuit applies $k$ Hadamard gates, optionally $k$ local $R_Z$ gates when $\beta\ne0$, and one $k$-qubit Pauli-$Z$ rotation, where $R_Z$ denotes a single-qubit rotation about the $Z$ axis. A standard controlled-NOT (CNOT) ladder decomposition of the Pauli rotation uses $2(k-1)$ logical CNOTs and one further $R_Z$ gate. Let $N_H$, $N_{R_Z}$, and $N_{\mathrm{CNOT}}$ denote the total numbers of the corresponding logical gates across all $B$ block preparations for one observation. Then
\begin{equation}
    N_H=Bk=d,\qquad
    N_{R_Z}=B+Bk\,\mathbf{1}_{\{\beta\ne0\}},\qquad
    N_{\mathrm{CNOT}}=2B(k-1).
    \label{eq:circuit-resource-count-qk}
\end{equation}
Here, $\mathbf{1}_{\{\beta\ne0\}}$ is the indicator that the one-body phase is active.
With the ladder executed serially within a block, the logical depth of one block-state preparation is $2k$ for $\beta=0$ and $2k+1$ for $\beta\ne0$, before hardware routing and basis-gate translation.

\begin{table}[htbp]
\centering
\caption{Logical circuit resources per observation. Depth refers to one block circuit under a CNOT-ladder decomposition.}
\label{tab:representation-scaling-qk}
\resizebox{\textwidth}{!}{%
\begin{tabular}{lrrrrrrr}
\toprule
Dataset/setting & $B$ & Peak qubits & InterFea dim. & $H$ & $R_Z$: $\beta=0$ ($\ne0$) & CNOTs & Block depth: $\beta=0$ ($\ne0$) \\
\midrule
Synthetic: $d=9,\ k=3$  & 3  & 3 & 12 & 9  & 3 (12)  & 12  & 6 (7) \\
Synthetic: $d=20,\ k=4$ & 5  & 4 & 25 & 20 & 5 (25)  & 30  & 8 (9) \\
Synthetic scaling: $d=60,\ k=6$ & 10 & 6 & 70 & 60 & 10 (70) & 100 & 12 (13) \\
Synthetic scaling: $d=80,\ k=8$ & 10 & 8 & 90 & 80 & 10 (90) & 140 & 16 (17) \\
Credit Card: $d=28,\ k=4$ & 7 & 4 & 35 & 28 & 7 (35) & 42 & 8 (9) \\
IEEE-CIS: $d=60,\ k=6$ & 10 & 6 & 70 & 60 & 10 (70) & 100 & 12 (13) \\
\bottomrule
\end{tabular}
}
\end{table}

In the higher-order scaling study, the block count is fixed at $B=10$. Increasing $k$ from six to eight therefore raises the peak circuit width from six to eight reusable qubits and the logical CNOT count per observation from $100$ to $140$, while the engineered classical representation grows from $70$ to $90$ coordinates. The feature map thus provides a direct circuit realization of the prescribed interaction order without enumerating all candidate degree-$k$ subsets. Under the known fixed partition, however, both the circuit representation and the engineered classical representation remain efficiently evaluable; the comparison is statistical rather than an asymptotic runtime separation.

For statevector-based evaluation, each observation is represented and cached by $B$ block statevectors. The statevector storage therefore scales as $O(B2^k)$ per observation, while one pairwise kernel entry costs $O(B2^k)$ when computed from the cached block states. Because the block overlap also admits the closed form in Eq.~\ref{eq:block_overlap_closed-qk}, the same kernel entry can alternatively be evaluated in $O(Bk)$ time.

The selected quantum configurations for both real-data benchmarks have $\beta\ne0$. Consequently, the per-observation logical gate counts are 28 Hadamard gates, 35 $R_Z$ rotations, and 42 CNOTs for Credit Card Fraud, and 60 Hadamard gates, 70 $R_Z$ rotations, and 100 CNOTs for IEEE-CIS. Since the blocks are evaluated separately, the peak simulated device sizes remain four and six qubits, respectively.

At the dataset level, let $n_{\mathrm{train}}$ and $n_{\mathrm{test}}$ denote the numbers of training and test observations. One complete construction produces
\begin{equation}
    B(n_{\mathrm{train}}+n_{\mathrm{test}})
\end{equation}
block statevectors. Since each statevector contains $2^k$ complex amplitudes, their temporary storage scales as $O\!\left(B2^k(n_{\mathrm{train}}+n_{\mathrm{test}})\right)$. The dense training and test--training kernel matrices require
\begin{equation}
    O\!\left(n_{\mathrm{train}}^2
    +n_{\mathrm{test}}n_{\mathrm{train}}\right)
\end{equation}
storage.

For one complete quantum-kernel construction using the selected hyperparameters, the Credit Card Fraud experiment contains $834$ training and $56{,}962$ test observations. Its temporary complex block-state arrays require approximately 98.8 MiB, while the corresponding float64 training and test--training kernel matrices together require approximately 367.8 MiB. The IEEE-CIS experiment contains $6000$ training and $29{,}527$ test observations, requiring approximately 346.9 MiB for the temporary complex block-state arrays and $1{,}626.3$ MiB for the corresponding float64 kernel matrices. These figures refer to the post-selection construction used for test evaluation. The complete test partitions are not used during model selection; hyperparameter tuning uses only the training-derived subtraining and validation partitions described below.

Although the kernel is defined through fidelities between states prepared by a Pauli-string quantum feature map, its block-factorized structure permits exact classical evaluation in the present setting. The reported gate counts therefore characterize a pre-transpilation logical circuit representation; they are not hardware-executed gate counts and should not be interpreted as evidence of quantum computational speedup. The hardware implementation used in Section~\ref{subsec:hardware-results-qk} is described in Section~\ref{subsec:hardware-protocol-qk}.

\section{Experimental Design}
\label{sec:experimentaldesign-qk}

\subsection{Synthetic settings}
\label{subsec:syntheticsettings-qk}

We evaluate two representative block-structured configurations:
\begin{itemize}[leftmargin=1.6em]
    \item Setting 1: $d=9$, $k=3$, $B=3$;
    \item Setting 2: $d=20$, $k=4$, $B=5$.
\end{itemize}
For both settings, we vary slab thickness $\epsilon\in\{0.05,0.15\}$ and training-label flip probability $p_{\mathrm{flip}}\in\{0.00,0.10\}$. Clean labels are generated from the block-regime rule, and label flips are applied to the training labels only. The test labels remain clean, so $p_{\mathrm{flip}}$ measures robustness to noisy supervision rather than contamination of the evaluation target.

We study both balanced and imbalanced label constructions. In the balanced case, labels are defined by the block-score majority rule with threshold $\tau=0$. Since the block-regime bits are symmetric under the synthetic sampling scheme and the number of blocks is odd in both settings, the clean positive-class proportion is theoretically $50\%$.

In the imbalanced case, labels are generated by the threshold rule in~\eqref{eq:global_label-qk} with $\tau=3$, making the positive class rarer. For $(d,k,B)=(9,3,3)$, this threshold requires all three block-regime bits to be positive, and hence the clean positive-class proportion is
\[
    \Pr(y=+1)=2^{-3}=12.5\%.
\]
For $(d,k,B)=(20,4,5)$, the threshold requires at least four of the five block-regime bits to be positive, giving
\[
    \Pr(y=+1)
    =
    \frac{\binom{5}{4}+\binom{5}{5}}{2^5}
    =
    18.75\%.
\]
These proportions refer to the clean labels. Label flips are applied only to the training split, before the subtrain/validation split, while the test labels remain clean. Therefore, when the clean positive-class proportion is $q$ and the flip probability is $p_{\mathrm{flip}}$, the expected noisy positive-class proportion in the training and validation portions is
\[
    q_{\mathrm{noisy}}
    =
    q(1-p_{\mathrm{flip}})+(1-q)p_{\mathrm{flip}}.
\]
For $p_{\mathrm{flip}}=0.10$, this gives an expected positive-class proportion of $50\%$ in the balanced setting, $20\%$ in the imbalanced $B=3$ setting, and $25\%$ in the imbalanced $B=5$ setting.

For every synthetic run, the available training set contains $2000$ observations. Of these, $1600$ are used for subtraining and $400$ for validation during model selection; the selected model is subsequently refitted on all $2000$ observations. Evaluation uses an independent test set of $500$ observations. Each configuration is evaluated over ten random seeds. For each setting and each random seed, hyperparameters are retuned from scratch on the validation split. Balanced experiments are selected by validation accuracy and reported by test accuracy. Imbalanced experiments use class-balanced sample weights during training, are selected by validation F1, and are reported by test F1.

\subsubsection{Higher-order scaling protocol}
\label{subsubsec:scaling-protocol-qk}

The higher-order scaling study evaluates whether the interaction-aligned quantum geometry remains effective as dimension and interaction order increase. It considers two imbalanced synthetic settings,
\begin{equation}
    (d,k,B)=(60,6,10)
    \qquad\text{and}\qquad
    (d,k,B)=(80,8,10).
\end{equation}
Holding $B=10$ fixed isolates the increase in the order of each planted interaction from an increase in the number of planted blocks. Both settings use $\epsilon=0.15$, $p_{\mathrm{flip}}=0$, and the same threshold $\tau=3$. Because the block score must exceed three, at least seven of the ten regime bits must be positive; hence the clean positive-class probability is
\begin{equation}
    \Pr(y=+1)
    =2^{-10}\sum_{j=7}^{10}\binom{10}{j}
    =\frac{176}{1024}
    \approx 17.19\%.
    \label{eq:scaling-positive-rate-qk}
\end{equation}

Each scaling setting uses $2000$ training observations, including $400$ validation observations for model selection, and an independent test set of $500$ observations. Results are reported as mean $\pm$ standard deviation over ten random seeds. All methods use class-balanced sample weights, are selected by validation F1, and are refitted on the full training sample before test evaluation. To match the other imbalanced synthetic experiments, we report test accuracy and F1, with F1 treated as the primary metric because accuracy can be inflated by majority-class predictions.

The baseline set contains linear, RBF, and Laplacian kernels; polynomial kernels of degrees $2$, $4$, $6$, and $8$; and \texttt{InterFea}, which is supplied with the ten planted block products. All method-specific hyperparameters are retuned for the higher-dimensional settings under the common validation protocol described below. The scale-aware search design is documented in Appendix~\ref{secA2-qk}.

\subsection{Baselines and model selection}
\label{subsec:baselines-qk}

We compare the quantum kernel against the following baselines under the same validation protocol:
\begin{itemize}[leftmargin=1.6em]
    \item linear kernel;
    \item RBF kernel;
    \item Laplacian kernel;
    \item polynomial kernels with degrees $2$, $4$, $6$, and $8$;
    \item engineered-interaction baseline, denoted \texttt{InterFea}.
\end{itemize}
The \texttt{InterFea} baseline uses the engineered feature map in~\eqref{eq:engineered_features-qk} and the regularized linear classifier described above. Its regularization parameter is selected on the validation subset. This deliberately strong baseline is given direct access to the planted interaction variables.

All kernel methods use kernel ridge regression with precomputed Gram matrices. For every setting and random seed, kernel and regularization hyperparameters are selected independently on the validation subset. For each polynomial degree, $\gamma$, $c$, and the kernel-ridge regularization parameter are tuned jointly. For the quantum kernel, $\alpha$, $\beta$, and the kernel-ridge regularization parameter are selected under the same validation protocol. The final selected model is then refitted using the full available training set.

\subsection{Hardware protocol}
\label{subsec:hardware-protocol-qk}

The hardware study uses the balanced and imbalanced $(d,k,B)=(20,4,5)$ settings with $\epsilon=0.15$ and $p_{\mathrm{flip}}\in\{0.00,0.10\}$, evaluated over ten seeds. Data generation, splits, label constructions, class weighting, and selection criteria are identical to those in Section~\ref{subsec:syntheticsettings-qk}.

Estimating \eqref{eq:q_fidelity_kernel-qk} with compute--uncompute circuits would require one circuit for every pair of observations. We avoid this by setting $\beta=0$, for which each block state reduces to the effective two-level state $\cos\varphi_b|{+}\rangle^{\otimes k}+i\sin\varphi_b|{-}\rangle^{\otimes k}$ with $\varphi_b=\alpha m_b(x)$, and the block kernel can be written as
\begin{equation}
    k_Q^{(b)}(x^{(b)},x'^{(b)})=\tfrac{1}{2}\bigl(1+\mathbf{r}_b(x)\cdot\mathbf{r}_b(x')\bigr),
    \label{eq:hardware-kernel-qk}
\end{equation}
where
\begin{equation}
    \mathbf{r}_b(x)=\bigl(\langle X_{b,1}\rangle,\ \langle Y_{b,1}Z_{b,2}\cdots Z_{b,k}\rangle\bigr)=\bigl(\cos 2\varphi_b,\,-\sin 2\varphi_b\bigr).
    \label{eq:hardware-bloch-qk}
\end{equation}
Both identities are derived in Appendix~\ref{secA3-qk}. Each observation therefore requires only two measurement settings, and every kernel entry is assembled classically; numerically, \eqref{eq:hardware-kernel-qk} agrees with the statevector fidelity kernel to within $10^{-13}$. On the simulator, fixing $\beta=0$ leaves the mean balanced accuracy unchanged and changes the mean imbalanced F1 by at most $0.003$ (not shown).

For each seed, $\alpha$ is selected on the simulator under the validation protocol of Section~\ref{subsec:baselines-qk} for the balanced and imbalanced labels with $p_{\mathrm{flip}}=0$, and the $p_{\mathrm{flip}}=0.10$ configurations reuse the same hardware runs.
The test set is not used in these choices. The kernel-ridge regularization parameter is then re-selected on the validation split using the hardware-estimated kernel, and the selected model is refitted on the full training set.

The experiments were run on \texttt{ibm\_kingston}, a 156-qubit IBM Heron~r2 processor, through Qiskit Runtime \citep{javadiabhari2024qiskit}. Each block is mapped to a chain of four connected qubits and implemented with the CNOT-ladder decomposition of Section~\ref{subsec:representation-efficiency-qk}, which compiles to six CZ gates without routing. Excluding qubits with readout error above $3\%$ and couplers with CZ error above $1\%$ leaves $22$ disjoint chains, which are evaluated in parallel within each circuit. Each block instance is measured in both settings ($X$ basis on all qubits; $Y$ basis on the first qubit and $Z$ basis on the others) with $512$ shots each, half of which apply an $X$ gate before measurement that is undone in post-processing to symmetrize readout errors; no other error mitigation is applied.
In total, the hardware runs used $27{,}312$ circuits and $2{,}028$ seconds of quantum processing time.

\subsection{Real-data benchmarks and evaluation protocol}
\label{subsec:realdatasets-qk}

We also consider two fraud-detection benchmarks: the Credit Card Fraud Detection dataset and the IEEE-CIS Fraud Detection dataset, both used through the Fraud Dataset Benchmark (FDB) \citep{grover2022fraud}. Both are highly imbalanced binary classification tasks. We use the complete test partition returned by FDB for each benchmark and report its positive-class proportion in the corresponding result subsection. This is particularly important because resampling is applied only to the training subset, while the test distribution is never resampled.

For both real-data benchmarks, we report accuracy and F1, with validation F1 used as the tuning objective because the positive class is rare and accuracy alone can be misleading. All real-data experiments follow a split-first evaluation protocol. The train/test split is performed before any resampling, standardization, feature selection, or block construction. The validation subset is drawn only from the training portion. Standardization parameters, feature-selection scores, and correlation-based groupings used to define blocks are computed only from the training portion and then applied to validation and test data. Resampling is applied exclusively to the model-fitting subset. The test distribution is never resampled, so the reported metrics reflect performance under the original class imbalance.

\section{Results}
\label{sec:results-qk}

\subsection{Balanced synthetic setting}
\label{subsec:balancedsynthetic-qk}

Tables~\ref{tab:d9k3_results-qk} and~\ref{tab:d20k4_results-qk} report balanced synthetic results as mean $\pm$ standard deviation over ten random seeds. In the smaller setting $(d,k,B)=(9,3,3)$, the quantum kernel achieves the best mean accuracy in all four balanced configurations, although the engineered-interaction baseline \texttt{InterFea} remains very close. This indicates that once the relevant block products are exposed directly to a classical model, much of the performance gain can already be recovered.

In the larger setting $(d,k,B)=(20,4,5)$, the pattern becomes clearer. Standard classical kernels achieve mean accuracies near the $0.5$ random-classification baseline, whereas both the quantum kernel and \texttt{InterFea} achieve substantially higher accuracy. The quantum kernel is best in all four balanced $d=20$ settings, with modest but consistent gains over \texttt{InterFea}. This supports the interpretation that interaction-aware geometry is the main driver of performance.

\begin{table}[htbp]
\centering
\caption{Balanced synthetic results for $(d,k,B)=(9,3,3)$, reported as mean test accuracy $\pm$ standard deviation over ten seeds. Columns correspond to $(\epsilon,p_{\mathrm{flip}})$.}
\label{tab:d9k3_results-qk}
\small
\setlength{\tabcolsep}{5pt}
\begin{tabular}{lcccc}
\toprule
Model 
& $(0.05,0.00)$
& $(0.05,0.10)$
& $(0.15,0.00)$
& $(0.15,0.10)$ \\
\midrule
Linear 
& $0.506 \pm 0.025$ 
& $0.494 \pm 0.020$ 
& $0.502 \pm 0.023$ 
& $0.504 \pm 0.019$ \\
RBF 
& $0.494 \pm 0.027$ 
& $0.497 \pm 0.032$ 
& $0.608 \pm 0.020$ 
& $0.593 \pm 0.038$ \\
Laplacian 
& $0.508 \pm 0.031$ 
& $0.510 \pm 0.024$ 
& $0.545 \pm 0.035$ 
& $0.538 \pm 0.026$ \\
Poly2 
& $0.490 \pm 0.016$ 
& $0.504 \pm 0.019$ 
& $0.490 \pm 0.016$ 
& $0.504 \pm 0.011$ \\
Poly4 
& $0.673 \pm 0.021$ 
& $0.641 \pm 0.027$ 
& $0.672 \pm 0.018$ 
& $0.644 \pm 0.023$ \\
Poly6 
& $0.593 \pm 0.021$ 
& $0.563 \pm 0.027$ 
& $0.639 \pm 0.018$ 
& $0.619 \pm 0.029$ \\
Poly8 
& $0.562 \pm 0.038$ 
& $0.549 \pm 0.028$ 
& $0.566 \pm 0.043$ 
& $0.550 \pm 0.038$ \\
InterFea 
& $0.765 \pm 0.025$ 
& $0.760 \pm 0.023$ 
& $0.767 \pm 0.022$ 
& $0.763 \pm 0.025$ \\
Quantum 
& $\mathbf{0.772 \pm 0.034}$ 
& $\mathbf{0.766 \pm 0.038}$ 
& $\mathbf{0.772 \pm 0.035}$ 
& $\mathbf{0.769 \pm 0.027}$ \\
\bottomrule
\end{tabular}
\end{table}

\begin{table}[htbp]
\centering
\caption{Balanced synthetic results for $(d,k,B)=(20,4,5)$, reported as mean test accuracy $\pm$ standard deviation over ten seeds. Columns correspond to $(\epsilon,p_{\mathrm{flip}})$.}
\label{tab:d20k4_results-qk}
\small
\setlength{\tabcolsep}{5pt}
\begin{tabular}{lcccc}
\toprule
Model 
& $(0.05,0.00)$
& $(0.05,0.10)$
& $(0.15,0.00)$
& $(0.15,0.10)$ \\
\midrule
Linear 
& $0.520 \pm 0.019$ 
& $0.509 \pm 0.018$ 
& $0.515 \pm 0.015$ 
& $0.513 \pm 0.018$ \\
RBF 
& $0.504 \pm 0.023$ 
& $0.504 \pm 0.007$ 
& $0.512 \pm 0.020$ 
& $0.502 \pm 0.022$ \\
Laplacian 
& $0.502 \pm 0.011$ 
& $0.510 \pm 0.019$ 
& $0.508 \pm 0.012$ 
& $0.505 \pm 0.015$ \\
Poly2 
& $0.512 \pm 0.024$ 
& $0.514 \pm 0.019$ 
& $0.507 \pm 0.018$ 
& $0.506 \pm 0.022$ \\
Poly4 
& $0.516 \pm 0.020$ 
& $0.496 \pm 0.005$ 
& $0.502 \pm 0.014$ 
& $0.487 \pm 0.025$ \\
Poly6 
& $0.492 \pm 0.029$ 
& $0.501 \pm 0.017$ 
& $0.499 \pm 0.020$ 
& $0.488 \pm 0.026$ \\
Poly8 
& $0.507 \pm 0.022$ 
& $0.494 \pm 0.021$ 
& $0.514 \pm 0.022$ 
& $0.495 \pm 0.021$ \\
InterFea 
& $0.733 \pm 0.029$ 
& $0.722 \pm 0.022$ 
& $0.734 \pm 0.027$ 
& $0.721 \pm 0.025$ \\
Quantum 
& $\mathbf{0.739 \pm 0.023}$ 
& $\mathbf{0.742 \pm 0.022}$ 
& $\mathbf{0.753 \pm 0.016}$ 
& $\mathbf{0.752 \pm 0.017}$ \\
\bottomrule
\end{tabular}
\end{table}

\subsubsection{Learning-curve analysis}
\label{subsubsec:learning-curves-qk}

To examine sample efficiency, we conduct an additional learning-curve analysis for the balanced $(d,k,B)=(20,4,5)$ setting with $\epsilon=0.15$ and $p_{\mathrm{flip}}=0$. The available training sample size is varied over $N\in\{250,500,1000,2000\}$. For each value of $N$, $20\%$ of the available observations are used for validation and hyperparameter selection, after which the selected model is refitted using all $N$ observations. Within each seed, the training sets are nested and evaluation uses the same fixed test set of 500 observations. The analysis follows the same ten-seed protocol as the corresponding synthetic experiment.

Figure~\ref{fig:learning-curves-d20-qk} shows that the quantum kernel has the highest mean accuracy at every investigated sample size. Its mean test accuracy increases from $0.718$ at $N=250$ to $0.753$ at $N=2000$, while its variability decreases as more training data become available. The engineered-interaction baseline also improves steadily, from $0.666$ to $0.734$, and therefore narrows the gap with the quantum kernel at the largest sample size. Linear, RBF, Laplacian, and polynomial kernels remain near the $0.5$ random-classification baseline throughout the investigated range.

\begin{figure}[htbp]
    \centering
    \includegraphics[width=\textwidth]{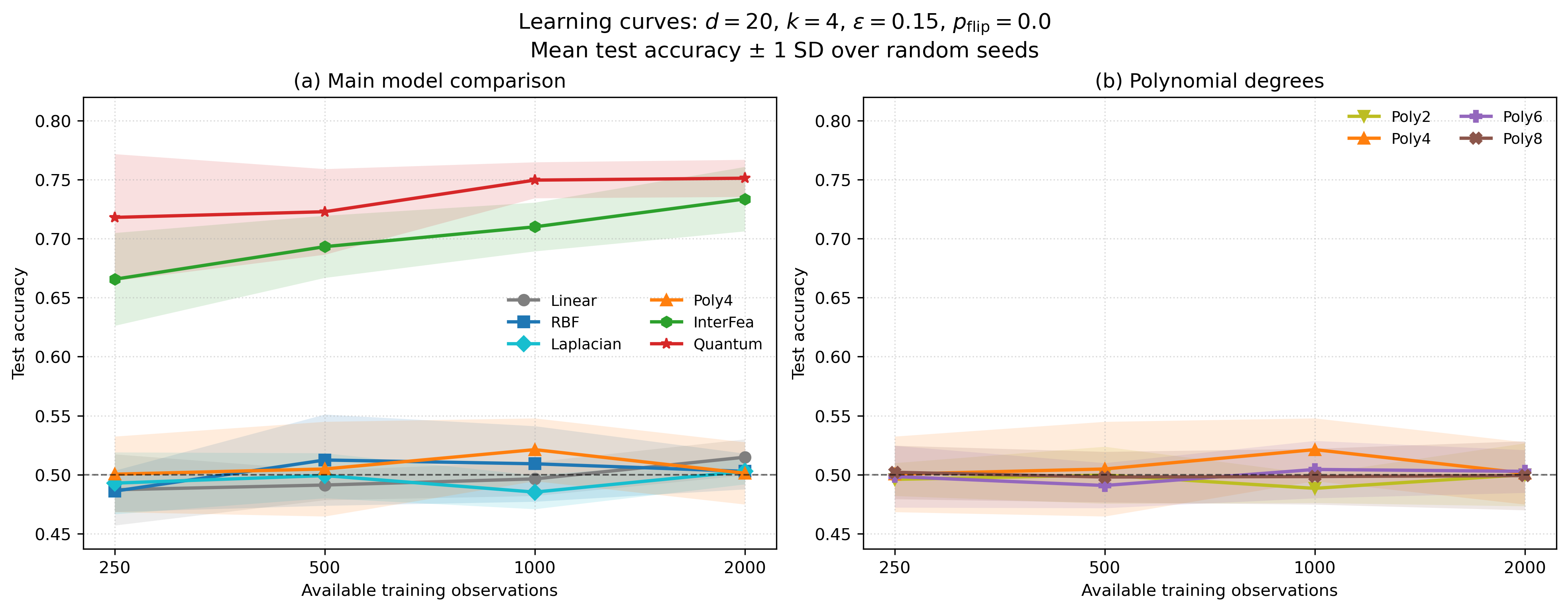}
    \caption{Learning curves for the balanced $(d,k)=(20,4)$ setting. Lines and bands show mean test accuracy and $\pm 1$ standard deviation.}
    \label{fig:learning-curves-d20-qk}
\end{figure}

The degree comparison in Figure~\ref{fig:learning-curves-d20-qk}(b) shows that polynomial kernels of degrees $2$, $4$, $6$, and $8$ remain near the $0.5$ random-classification baseline as the training sample size increases. Although the higher-degree polynomial kernels provide a richer generic interaction space, increasing either polynomial degree or sample size alone does not isolate the sparse planted interaction signal over the investigated range. By contrast, both structure-aligned representations improve with additional data, and the quantum kernel maintains the highest mean accuracy at every sample size. This pattern supports the interpretation that alignment with the planted interaction structure, rather than generic nonlinear capacity, improves sample efficiency.

\subsection{Imbalanced synthetic setting}
\label{subsec:imbalancedsynthetic-qk}

For the imbalanced synthetic experiments, we use test F1 as the main metric, with model selection based on validation F1. Tables~\ref{tab:d9k3_imbalanced_metrics-qk} and~\ref{tab:d20k4_imbalanced_metrics-qk} summarize the results. In both settings, the quantum kernel achieves the best mean F1 in all four configurations, while \texttt{InterFea} remains the closest competitor. This reinforces the conclusion that the gain is tied to interaction-aware similarity rather than generic nonlinear modeling.

\begin{table}[htbp]
\centering
\caption{Imbalanced synthetic results for $(d,k,B)=(9,3,3)$: mean test accuracy and F1 $\pm$ standard deviation over ten seeds. Columns correspond to $(\epsilon,p_{\mathrm{flip}})$. Bold and underlined values denote the best and second-best results.}
\label{tab:d9k3_imbalanced_metrics-qk}
\small
\setlength{\tabcolsep}{4pt}
\resizebox{\textwidth}{!}{%
\begin{tabular}{lcccccccc}
\toprule
\multirow{2}{*}{Model}
& \multicolumn{2}{c}{$(0.05,0.00)$}
& \multicolumn{2}{c}{$(0.05,0.10)$}
& \multicolumn{2}{c}{$(0.15,0.00)$}
& \multicolumn{2}{c}{$(0.15,0.10)$} \\
\cmidrule(lr){2-3}
\cmidrule(lr){4-5}
\cmidrule(lr){6-7}
\cmidrule(lr){8-9}
& Acc. & F1 & Acc. & F1 & Acc. & F1 & Acc. & F1 \\
\midrule
Linear
& $0.506 \pm 0.012$ & $0.200 \pm 0.027$
& $0.497 \pm 0.025$ & $0.194 \pm 0.036$
& $0.499 \pm 0.014$ & $0.197 \pm 0.033$
& $0.493 \pm 0.021$ & $0.191 \pm 0.040$ \\
RBF
& $0.588 \pm 0.032$ & $0.179 \pm 0.041$
& $0.568 \pm 0.019$ & $0.191 \pm 0.032$
& $0.748 \pm 0.050$ & $0.268 \pm 0.065$
& $0.608 \pm 0.079$ & $0.215 \pm 0.039$ \\
Laplacian
& $\mathbf{0.834 \pm 0.024}$ & $0.101 \pm 0.054$
& $\mathbf{0.773 \pm 0.018}$ & $0.161 \pm 0.029$
& $\mathbf{0.848 \pm 0.021}$ & $0.106 \pm 0.040$
& $\mathbf{0.791 \pm 0.016}$ & $0.192 \pm 0.031$ \\
Poly2
& $0.539 \pm 0.028$ & $0.180 \pm 0.037$
& $0.522 \pm 0.015$ & $0.180 \pm 0.026$
& $0.533 \pm 0.024$ & $0.181 \pm 0.039$
& $0.522 \pm 0.021$ & $0.183 \pm 0.031$ \\
Poly4
& $0.739 \pm 0.030$ & $0.320 \pm 0.082$
& $0.660 \pm 0.021$ & $0.290 \pm 0.039$
& $0.752 \pm 0.026$ & $0.351 \pm 0.066$
& $0.654 \pm 0.041$ & $0.277 \pm 0.051$ \\
Poly6
& $0.678 \pm 0.053$ & $0.237 \pm 0.040$
& $0.637 \pm 0.028$ & $0.242 \pm 0.049$
& $0.735 \pm 0.030$ & $0.291 \pm 0.056$
& $0.662 \pm 0.045$ & $0.263 \pm 0.060$ \\
Poly8
& $0.627 \pm 0.052$ & $0.225 \pm 0.044$
& $0.609 \pm 0.053$ & $0.204 \pm 0.048$
& $0.676 \pm 0.032$ & $0.245 \pm 0.046$
& $0.596 \pm 0.010$ & $0.219 \pm 0.026$ \\
InterFea
& $0.801 \pm 0.039$ & $\underline{0.487 \pm 0.067}$
& $0.719 \pm 0.047$ & $\underline{0.438 \pm 0.077}$
& $0.803 \pm 0.036$ & $\underline{0.490 \pm 0.065}$
& $0.722 \pm 0.049$ & $\underline{0.442 \pm 0.079}$ \\
Quantum
& $\underline{0.806 \pm 0.035}$ & $\mathbf{0.494 \pm 0.070}$
& $\underline{0.725 \pm 0.040}$ & $\mathbf{0.446 \pm 0.066}$
& $\underline{0.827 \pm 0.037}$ & $\mathbf{0.544 \pm 0.087}$
& $\underline{0.730 \pm 0.042}$ & $\mathbf{0.449 \pm 0.067}$ \\
\bottomrule
\end{tabular}%
}
\end{table}

\begin{table}[htbp]
\centering
\caption{Imbalanced synthetic results for $(d,k,B)=(20,4,5)$: mean test accuracy and F1 $\pm$ standard deviation over ten seeds. Columns correspond to $(\epsilon,p_{\mathrm{flip}})$. Bold and underlined values denote the best and second-best results.}
\label{tab:d20k4_imbalanced_metrics-qk}
\small
\setlength{\tabcolsep}{4pt}
\resizebox{\textwidth}{!}{%
\begin{tabular}{lcccccccc}
\toprule
\multirow{2}{*}{Model}
& \multicolumn{2}{c}{$(0.05,0.00)$}
& \multicolumn{2}{c}{$(0.05,0.10)$}
& \multicolumn{2}{c}{$(0.15,0.00)$}
& \multicolumn{2}{c}{$(0.15,0.10)$} \\
\cmidrule(lr){2-3}
\cmidrule(lr){4-5}
\cmidrule(lr){6-7}
\cmidrule(lr){8-9}
& Acc. & F1 & Acc. & F1 & Acc. & F1 & Acc. & F1 \\
\midrule
Linear
& $0.502 \pm 0.038$ & $0.281 \pm 0.044$
& $0.514 \pm 0.026$ & $0.296 \pm 0.032$
& $0.506 \pm 0.036$ & $0.284 \pm 0.050$
& $0.511 \pm 0.028$ & $0.289 \pm 0.041$ \\
RBF
& $0.646 \pm 0.023$ & $0.207 \pm 0.033$
& $0.619 \pm 0.012$ & $0.242 \pm 0.023$
& $0.678 \pm 0.030$ & $0.193 \pm 0.040$
& $0.672 \pm 0.030$ & $0.221 \pm 0.047$ \\
Laplacian
& $0.728 \pm 0.023$ & $0.188 \pm 0.024$
& $0.691 \pm 0.019$ & $0.211 \pm 0.059$
& $0.735 \pm 0.018$ & $0.183 \pm 0.058$
& $0.700 \pm 0.016$ & $0.205 \pm 0.048$ \\
Poly2
& $0.573 \pm 0.019$ & $0.273 \pm 0.034$
& $0.542 \pm 0.019$ & $0.267 \pm 0.019$
& $0.568 \pm 0.026$ & $0.258 \pm 0.031$
& $0.547 \pm 0.006$ & $0.271 \pm 0.019$ \\
Poly4
& $0.592 \pm 0.026$ & $0.237 \pm 0.039$
& $0.581 \pm 0.034$ & $0.233 \pm 0.029$
& $0.618 \pm 0.021$ & $0.221 \pm 0.041$
& $0.586 \pm 0.028$ & $0.229 \pm 0.035$ \\
Poly6
& $0.587 \pm 0.036$ & $0.249 \pm 0.013$
& $0.591 \pm 0.034$ & $0.231 \pm 0.018$
& $0.649 \pm 0.033$ & $0.197 \pm 0.040$
& $0.597 \pm 0.028$ & $0.215 \pm 0.054$ \\
Poly8
& $0.602 \pm 0.052$ & $0.241 \pm 0.032$
& $0.582 \pm 0.046$ & $0.252 \pm 0.037$
& $0.640 \pm 0.025$ & $0.212 \pm 0.051$
& $0.629 \pm 0.011$ & $0.225 \pm 0.054$ \\
InterFea
& $\underline{0.756 \pm 0.012}$ & $\underline{0.508 \pm 0.030}$
& $\underline{0.714 \pm 0.018}$ & $\underline{0.493 \pm 0.025}$
& $\underline{0.753 \pm 0.017}$ & $\underline{0.505 \pm 0.036}$
& $\underline{0.710 \pm 0.022}$ & $\underline{0.491 \pm 0.026}$ \\
Quantum
& $\mathbf{0.762 \pm 0.016}$ & $\mathbf{0.525 \pm 0.034}$
& $\mathbf{0.727 \pm 0.020}$ & $\mathbf{0.514 \pm 0.031}$
& $\mathbf{0.770 \pm 0.016}$ & $\mathbf{0.534 \pm 0.029}$
& $\mathbf{0.732 \pm 0.017}$ & $\mathbf{0.516 \pm 0.025}$ \\
\bottomrule
\end{tabular}%
}
\end{table}

The training-label flip experiments test resilience to noisy supervision. At $p_{\mathrm{flip}}=0.10$, the quantum kernel retains the highest mean accuracy in every balanced configuration and the highest mean F1 in every imbalanced configuration. Balanced accuracy changes only slightly, from a decrease of $0.006$ to an increase of $0.003$, whereas imbalanced F1 decreases by $0.011$--$0.095$, reflecting the greater sensitivity of minority-class learning to label corruption. These results demonstrate comparative resilience to $10\%$ training-label noise within the evaluated settings, rather than general robustness to arbitrary or quantum-device noise.

\subsubsection{Scaling to higher interaction order and dimension}
\label{subsubsec:scaling-results-qk}

Table~\ref{tab:scaling-k6-k8-qk} presents the higher-order scaling results. The quantum kernel achieves the highest mean accuracy and F1 in both settings, so its advantage is not confined to either overall classification performance or minority-class recovery. At $(d,k,B)=(60,6,10)$, it obtains an accuracy of $0.706$ and an F1 of $0.412$. The corresponding \texttt{InterFea} values are $0.672$ and $0.379$, while the strongest generic classical results across the two metrics are an accuracy of $0.651$ and an F1 of $0.247$. At $(d,k,B)=(80,8,10)$, the quantum kernel retains an accuracy of $0.695$ and an F1 of $0.395$, compared with $0.631$ and $0.320$ for \texttt{InterFea}; the strongest generic classical results are $0.661$ in accuracy and $0.256$ in F1. Relative to the strongest generic result for each metric, the quantum kernel improves mean accuracy by approximately $8\%$ and $5\%$, and mean F1 by approximately $67\%$ and $54\%$, at $k=6$ and $k=8$, respectively.

\begin{table}[htbp]
\centering
\caption{Higher-order imbalanced scaling results at $\epsilon=0.15$ and $p_{\mathrm{flip}}=0$: mean test accuracy and F1 $\pm$ standard deviation over ten seeds. Bold and underlined values denote the best and second-best results within each column.}
\label{tab:scaling-k6-k8-qk}
\small
\setlength{\tabcolsep}{5pt}
\begin{tabular}{lcccc}
\toprule
\multirow{2}{*}{Model}
& \multicolumn{2}{c}{$(d,k,B)=(60,6,10)$}
& \multicolumn{2}{c}{$(d,k,B)=(80,8,10)$} \\
\cmidrule(lr){2-3}\cmidrule(lr){4-5}
& Acc. & F1 & Acc. & F1 \\
\midrule
Linear
& $0.494 \pm 0.024$ & $0.247 \pm 0.017$
& $0.508 \pm 0.020$ & $0.252 \pm 0.043$ \\
RBF
& $0.597 \pm 0.046$ & $0.225 \pm 0.034$
& $0.580 \pm 0.058$ & $0.237 \pm 0.047$ \\
Laplacian
& $0.651 \pm 0.054$ & $0.200 \pm 0.043$
& $\underline{0.661 \pm 0.058}$ & $0.224 \pm 0.054$ \\
Poly2
& $0.574 \pm 0.036$ & $0.241 \pm 0.037$
& $0.551 \pm 0.067$ & $0.256 \pm 0.046$ \\
Poly4
& $0.603 \pm 0.079$ & $0.221 \pm 0.048$
& $0.558 \pm 0.054$ & $0.252 \pm 0.046$ \\
Poly6
& $0.607 \pm 0.055$ & $0.231 \pm 0.041$
& $0.584 \pm 0.054$ & $0.242 \pm 0.045$ \\
Poly8
& $0.612 \pm 0.054$ & $0.223 \pm 0.041$
& $0.581 \pm 0.040$ & $0.236 \pm 0.038$ \\
InterFea
& $\underline{0.672 \pm 0.019}$ & $\underline{0.379 \pm 0.048}$
& $0.631 \pm 0.026$ & $\underline{0.320 \pm 0.054}$ \\
Quantum
& $\mathbf{0.706 \pm 0.017}$ & $\mathbf{0.412 \pm 0.056}$
& $\mathbf{0.695 \pm 0.020}$ & $\mathbf{0.395 \pm 0.043}$ \\
\bottomrule
\end{tabular}
\end{table}

The agreement between the two metrics is important: the quantum kernel ranks first in accuracy and F1 in both scaling settings. Nevertheless, F1 remains the primary selection and interpretation metric because the task is imbalanced. The Laplacian kernel illustrates this distinction: it attains relatively high accuracies of $0.651$ and $0.661$ while producing F1 scores of only $0.200$ and $0.224$. Accuracy is therefore reported alongside F1 for consistency and to show that the quantum model's improved minority-class recovery is not obtained at the expense of overall classification performance.

Because all methods use the same random seeds, the comparison with \texttt{InterFea} can be evaluated pairwise. The quantum-minus-\texttt{InterFea} mean F1 difference is $0.032$ at $(d,k)=(60,6)$, with a $95\%$ paired $t$ confidence interval of $[0.011,0.053]$ ($p=0.007$, two-sided). At $(d,k)=(80,8)$, the difference is $0.075$, with an interval of $[0.032,0.119]$ ($p=0.004$). The fidelity kernel therefore adds predictive value beyond exposing the planted products to a linear classifier, with a larger gap at $k=8$.

Increasing polynomial degree does not close the gap. The degree-matched Poly6 and Poly8 kernels achieve mean F1 scores of only $0.231$ and $0.236$ in their respective settings. These results indicate that matching the nominal polynomial degree is insufficient: the generic polynomial kernels do not isolate the ten prescribed block interactions, whereas the proposed feature map concentrates its similarity geometry on those interactions. This comparison supports an inductive-bias explanation rather than an explicit-enumeration or computational-complexity claim.

Together with the controlled third- and fourth-order results, the scaling study establishes a persistent advantage across $k\in\{3,4,6,8\}$. The quantum kernel outperforms \texttt{InterFea} in every balanced and imbalanced synthetic configuration. At $\epsilon=0.15$ and $p_{\mathrm{flip}}=0$, the mean F1 gaps are approximately $0.054$, $0.029$, $0.032$, and $0.075$, respectively. The gap is positive at every order and largest at $k=8$. Thus, the contribution of the Pauli-string fidelity geometry is not order-specific and remains effective as input dimension and interaction order increase.

\subsection{Execution on IBM Quantum hardware}
\label{subsec:hardware-results-qk}

Table~\ref{tab:hardware-d20-qk} compares the hardware-estimated kernel with \texttt{InterFea} and with the noise-free simulator on the same ten seeds, following the protocol of Section~\ref{subsec:hardware-protocol-qk}.

\begin{table}[htbp]
\centering
\caption{IBM Quantum hardware results for $(d,k,B)=(20,4,5)$ at $\epsilon=0.15$: mean $\pm$ standard deviation over ten seeds. Balanced tasks report test accuracy; imbalanced tasks report test accuracy and F1. The simulator row is a noise-free simulation on the same seeds.}
\label{tab:hardware-d20-qk}
\small
\setlength{\tabcolsep}{4pt}
\resizebox{\textwidth}{!}{%
\begin{tabular}{lcccccc}
\toprule
& \multicolumn{2}{c}{Balanced: Acc.}
& \multicolumn{2}{c}{Imbalanced, $p_{\mathrm{flip}}=0.00$}
& \multicolumn{2}{c}{Imbalanced, $p_{\mathrm{flip}}=0.10$} \\
\cmidrule(lr){2-3}
\cmidrule(lr){4-5}
\cmidrule(lr){6-7}
Model & $p_{\mathrm{flip}}=0.00$ & $p_{\mathrm{flip}}=0.10$ & Acc. & F1 & Acc. & F1 \\
\midrule
InterFea
& $0.734 \pm 0.027$ & $0.721 \pm 0.025$
& $0.753 \pm 0.017$ & $0.505 \pm 0.036$
& $0.710 \pm 0.022$ & $0.491 \pm 0.026$ \\
Quantum, simulator
& $0.753 \pm 0.016$ & $0.752 \pm 0.017$
& $0.770 \pm 0.016$ & $0.534 \pm 0.029$
& $0.732 \pm 0.017$ & $0.516 \pm 0.025$ \\
Quantum, IBM hardware
& $0.749 \pm 0.013$ & $0.745 \pm 0.015$
& $0.760 \pm 0.017$ & $0.531 \pm 0.030$
& $0.728 \pm 0.021$ & $0.510 \pm 0.034$ \\
\bottomrule
\end{tabular}%
}
\end{table}

The hardware-estimated kernel closely reproduces the simulator. In each of the four label configurations, its mean primary metric (accuracy for balanced and F1 for imbalanced tasks) lies $0.003$--$0.007$ below the simulator value. This gap is well within the seed-to-seed standard deviations of $0.013$--$0.034$. The ordering in Tables~\ref{tab:d20k4_results-qk} and~\ref{tab:d20k4_imbalanced_metrics-qk} is preserved on hardware. The hardware-estimated kernel has a higher mean primary metric than \texttt{InterFea} in all four configurations. By contrast, the best generic classical kernel reaches a mean balanced accuracy of only $0.515$ and a mean imbalanced F1 below $0.29$.

The individual observables are also estimated accurately. In each hardware run, the measured $\langle Y_{b,1}Z_{b,2}Z_{b,3}Z_{b,4}\rangle$ has a correlation between $0.974$ and $0.990$ with its ideal value $-\sin 2\alpha m_b$ across all $12{,}500$ block instances. Device noise mainly shrinks the measured values towards zero, by about $14\%$ for the $Y$-type component, and the remaining scatter, with a standard deviation of about $0.044$, matches the shot-noise level $1/\sqrt{512}$. This uniform shrinkage largely preserves the structure of the kernel, which explains the small loss in predictive performance.

\subsection{Credit Card Fraud Detection benchmark}
\label{subsec:ccfraud-qk}

We first evaluate the models on the Credit Card Fraud Detection dataset from the Fraud Dataset Benchmark, originally associated with European cardholder transactions \citep{dalpozzolo2015calibrating,grover2022fraud}. The original dataset contains $284{,}807$ credit-card transactions over a two-day period, with $492$ fraudulent transactions and $284{,}315$ non-fraudulent transactions, corresponding to an overall positive-class fraud rate of approximately $0.173\%$. The original features are anonymized for confidentiality; in our experiment, we use the $28$ numerical input features and treat fraud detection as a highly imbalanced binary classification problem.

Following the chronological split used by the Fraud Dataset Benchmark, transactions are ordered by the \texttt{Time} variable. The earliest $80\%$ are assigned to training and the latest $20\%$ to testing. This produces an initial training set of $227{,}845$ observations and a fixed test set of $56{,}962$ observations. The training portion contains $417$ fraudulent transactions, while the test set contains $75$ fraudulent and $56{,}887$ non-fraudulent transactions, corresponding to a test fraud rate of approximately $0.132\%$.

Resampling is applied only after the chronological train--test split. In each run, all $417$ fraudulent training observations are retained and an equal number of non-fraudulent training observations are randomly sampled, producing a balanced model-fitting set of $834$ observations. A validation fraction of $0.2$ is used for hyperparameter selection, giving $667$ subtraining and $167$ validation observations. After model selection, the selected model is refitted on all $834$ resampled training observations and evaluated on the complete chronological test set. The resulting feature dimension is $d=28$, and the quantum kernel uses a block structure with $k=4$ and $B=7$.

The chronological train--test division is fixed across all runs. Ten seeds vary the sampled non-fraudulent training observations and the internal validation partition. The reported standard deviations therefore measure sensitivity to training-side resampling and validation selection rather than variation across different test sets. Since the positive class remains extremely rare in the test distribution, all models are tuned using validation F1.

The results in Table~\ref{tab:ccfraud_main_metrics-qk} show that the quantum kernel achieves the highest mean accuracy and F1 among the evaluated methods, while Poly2 is the second-best method on both metrics. The quantum kernel improves mean F1 from $0.348$ for Poly2 to $0.472$, corresponding to an absolute improvement of $0.124$ and a relative improvement of approximately $36\%$. It also achieves a higher F1 than Poly2 in each of the ten runs. Evaluation on the complete chronological test partition avoids an additional test-set subsampling step and preserves the original temporal test distribution.

\begin{table}[htbp]
\centering
\caption{Credit Card Fraud results on the complete chronological FDB test partition: mean test accuracy and F1 $\pm$ standard deviation over ten training-resampling seeds. Bold and underlined values denote the best and second-best results.}
\label{tab:ccfraud_main_metrics-qk}
\begin{tabular}{lcc}
\toprule
Model & Acc. & F1 \\
\midrule
Linear    & $0.983 \pm 0.007$ & $0.124 \pm 0.050$ \\
RBF       & $0.995 \pm 0.001$ & $0.335 \pm 0.085$ \\
Laplacian & $0.994 \pm 0.000$ & $0.275 \pm 0.013$ \\
Poly2     & $\underline{0.996 \pm 0.000}$ & $\underline{0.348 \pm 0.021}$ \\
Poly4     & $0.983 \pm 0.002$ & $0.119 \pm 0.015$ \\
Poly6     & $0.977 \pm 0.011$ & $0.105 \pm 0.035$ \\
Poly8     & $0.968 \pm 0.011$ & $0.073 \pm 0.022$ \\
InterFea  & $0.991 \pm 0.002$ & $0.197 \pm 0.034$ \\
Quantum   & $\mathbf{0.998 \pm 0.001}$ & $\mathbf{0.472 \pm 0.075}$ \\
\bottomrule
\end{tabular}
\end{table}

\subsection{IEEE-CIS Fraud Detection benchmark}
\label{subsec:ieeecis-qk}

We next evaluate the models on the IEEE-CIS Fraud Detection dataset from the Fraud Dataset Benchmark \citep{ieeecis2019fraud,grover2022fraud}. The original IEEE-CIS dataset is a large-scale transaction-fraud benchmark released through the IEEE-CIS Fraud Detection Kaggle competition and provided by Vesta Corporation. The FDB version uses a standardized chronological train--test split and has a training positive-class ratio of approximately $3.50\%$, indicating a highly imbalanced classification task.

Our experiment uses the numeric-only FDB version of IEEE-CIS. We select $60$ numerical features using correlation-based feature selection and use a block structure with $d=60$, $k=6$, and $B=10$ for the quantum kernel. The feature blocks are constructed by correlation-sorted grouping. We use the complete train and test partitions returned by the FDB loader, without further subsampling of the test set. These partitions contain $561{,}013$ training observations and $29{,}527$ test observations. The test set contains $1{,}169$ fraudulent and $28{,}358$ non-fraudulent transactions, corresponding to a positive-class fraud rate of approximately $3.96\%$.

Resampling is applied only to the training partition. In each run, up to $3000$ fraud samples and the same number of non-fraud samples are used, producing a balanced training subset of size $6000$. A validation fraction of $0.2$ is used within this resampled training set for hyperparameter selection, giving $4800$ subtraining and $1200$ validation observations. After model selection, the selected model is refitted on all $6000$ training observations and evaluated on the fixed $29{,}527$-observation FDB test set. As in the Credit Card Fraud Detection experiment, all models are tuned using validation F1 and evaluated under the original class imbalance. The experiment is repeated over ten random seeds, and we report mean $\pm$ standard deviation.

The results are shown in Table~\ref{tab:ieeecis_main_metrics-qk}. The Laplacian kernel achieves the strongest mean accuracy and F1, with values of $0.831$ and $0.267$, respectively. The quantum kernel obtains the second-best result on both metrics, with accuracy $0.821$ and F1 $0.239$. Evaluation on the complete FDB test partition preserves the original class imbalance and avoids an additional test-set subsampling step. These results suggest that the proposed quantum feature map remains competitive on this large real tabular fraud benchmark, although its inductive bias appears less closely aligned with this dataset than with the controlled synthetic tasks.

\begin{table}[htbp]
\centering
\caption{IEEE-CIS results on the complete FDB test partition: mean test accuracy and F1 $\pm$ standard deviation over ten seeds. Bold and underlined values denote the best and second-best results.}
\label{tab:ieeecis_main_metrics-qk}
\begin{tabular}{lcc}
\toprule
Model & Acc. & F1 \\
\midrule
Linear    & $0.804 \pm 0.005$ & $0.216 \pm 0.004$ \\
RBF       & $0.812 \pm 0.007$ & $0.231 \pm 0.008$ \\
Laplacian & $\mathbf{0.831 \pm 0.008}$ & $\mathbf{0.267 \pm 0.007}$ \\
Poly2     & $0.806 \pm 0.008$ & $0.224 \pm 0.008$ \\
Poly4     & $0.809 \pm 0.007$ & $0.227 \pm 0.004$ \\
Poly6     & $0.815 \pm 0.011$ & $0.229 \pm 0.009$ \\
Poly8     & $0.812 \pm 0.010$ & $0.225 \pm 0.007$ \\
InterFea  & $0.803 \pm 0.005$ & $0.217 \pm 0.006$ \\
Quantum   & $\underline{0.821 \pm 0.006}$ & $\underline{0.239 \pm 0.008}$ \\
\bottomrule
\end{tabular}
\end{table}

\section{Discussion}
\label{sec:discussion-qk}

\subsection{Mechanism of improvement}

The synthetic results indicate that improvement comes primarily from alignment with the planted interaction structure rather than generic nonlinearity. Distance-based kernels remain weak where Euclidean closeness does not imply regime agreement, and degree-matched Poly6 and Poly8 remain substantially below the proposed kernel. \texttt{InterFea} is consistently the strongest classical competitor, confirming that direct access to the block products explains much of the gain. Yet the quantum kernel also outperforms \texttt{InterFea} in every synthetic configuration across $k\in\{3,4,6,8\}$. The product-of-fidelities geometry therefore adds predictive value beyond a linear representation of the planted variables.

\subsection{Quantum-specific contribution}

Across the staged synthetic design, the Pauli-string fidelity kernel outperforms both generic kernels and \texttt{InterFea} at the evaluated third-, fourth-, sixth-, and eighth-order interactions. In the sixth- and eighth-order scaling settings, it retains mean accuracy near $0.70$ and mean F1 near $0.40$, ranking first on both metrics. Degree-matched polynomial kernels remain near $0.23$ in F1, while \texttt{InterFea} falls from $0.379$ to $0.320$. These results provide direct scaling evidence that the quantum-induced similarity remains effective as input dimension and interaction order increase.

The quantum-specific contribution consists of two linked elements. The first is a block-factorized circuit-native realization in which a $k$-qubit Pauli-$Z$ string directly encodes each planted $k$-way product. The second is the nonlinear similarity produced by state fidelity and multiplicative composition across blocks. The persistent advantage over generic kernels shows the value of this quantum-induced inductive bias, while the advantage over \texttt{InterFea} shows that the observed improvement is not exhausted by placing the planted products in a classical linear feature vector. The scaling study therefore provides evidence for a predictively useful quantum-induced representation in the statistical and geometric sense.

This quantum-induced predictive advantage is distinct from computational quantum advantage. Under the known disjoint block partition used in these experiments, the fidelity kernel admits an exact classical formula and can be evaluated efficiently. The evidence therefore concerns predictive performance, interaction-aware geometry, and persistence across the investigated scaling settings rather than classical intractability. Demonstrating computational quantum advantage would require a non-factorizing or otherwise classically hard feature map and an explicit complexity separation, neither of which is assumed here. The hardware study in Section~\ref{subsec:hardware-results-qk} complements this picture: the kernel can be estimated on a current superconducting processor, and its predictive performance in the evaluated $d=20$ setting is retained within seed-to-seed variability.

\subsection{Real-data interpretation}

The real-data experiments add nuance. On Credit Card Fraud, the quantum kernel achieves the highest mean accuracy and F1 on the fixed chronological test set. On IEEE-CIS, the Laplacian kernel performs best, while the quantum kernel is second. This suggests that the usefulness of the quantum feature map depends on how closely the true predictive structure resembles the interaction geometry encoded by the circuit. Real fraud labels are not generated by the synthetic block-product rule, and their structure may involve temporal, categorical, behavioral, or graph-based effects absent from our numerical feature representation. Thus, the real-data results should be interpreted as evidence of competitiveness beyond the controlled synthetic setting, not as universal dominance.

\subsection{Limitations}

The main limitation of the controlled synthetic study is that the data-generating mechanism is deliberately aligned with the proposed feature map. In real data, the relevant interactions and block structure are not known in advance, and the IEEE-CIS result shows that the kernel is not uniformly dominant. Although the complete chronological Credit Card Fraud test partition contains $56{,}962$ observations, it includes only $75$ fraudulent transactions. Across both real-data benchmarks, the repetitions share fixed test partitions and vary the training-side resampling and internal validation split. Their standard deviations therefore quantify sensitivity to model fitting and selection, but not uncertainty across alternative temporal test periods. In addition, apart from the hardware study in Section~\ref{subsec:hardware-results-qk}, the reported experiments use exact numerical simulation. That study is itself limited to $\beta=0$, a single synthetic setting at $\epsilon=0.15$, ten seeds, and one processor, and it applies no error mitigation beyond readout symmetrization.

The scaling study demonstrates predictive persistence up to 80 dimensions and eighth-order interactions. The block-factorized construction extends to larger $d$ and $k$: direct statevector storage costs $O(B2^k)$ per observation, while the exact overlap in Eq.~\ref{eq:block_overlap_closed-qk} permits $O(Bk)$-time evaluation of each kernel entry. Larger values of $d$ and $k$ are therefore accessible without enumerating all degree-$k$ interactions, although dense kernel storage remains quadratic in sample size. Whether the predictive advantage persists at those scales requires further evaluation.

Future work should prioritize evaluation at larger dimensions and interaction orders, data-driven block discovery, trainable or overlapping Pauli-string feature maps, and hardware evaluation beyond the setting studied here, including $\beta\neq0$, higher interaction orders, and the real-data benchmarks. Establishing computational quantum advantage would additionally require a non-factorizing or otherwise classically hard feature map, together with explicit complexity and runtime comparisons against strong classical interaction-learning and kernel-approximation methods.

\section{Conclusion}
\label{sec:conclusion-qk}

This study examined quantum-kernel learning through the lens of inductive-bias alignment. We introduced a thin-slab interaction model in which labels are generated by sparse block-wise $k$-way product signs, creating a controlled setting where Euclidean distance and low-order moments are misaligned with the target structure. We then designed a block-factorized quantum feature map whose Pauli-string phases encode the relevant products directly, yielding a fidelity kernel with an explicit interaction-aware geometry.

The experiments show that quantum-kernel performance depends on feature-map--target alignment. From the controlled third- and fourth-order settings to the 60-dimensional sixth-order and 80-dimensional eighth-order settings, the proposed kernel consistently outperforms generic classical kernels and \texttt{InterFea}. Its advantage over \texttt{InterFea} is positive across all evaluated orders and largest at eighth order, showing that the product-of-fidelities geometry contributes beyond the planted products alone. On the fraud benchmarks, the method achieves the strongest F1 on Credit Card Fraud and ranks second on IEEE-CIS. Executed on a 156-qubit IBM Quantum processor, the kernel matches the noise-free simulator within seed-to-seed variability in the $(d,k,B)=(20,4,5)$ setting and retains its mean advantage over \texttt{InterFea} there.

These findings establish a quantum-induced predictive advantage in the interaction-driven regimes studied here: high-order Pauli-string phases provide a block-factorized circuit-native representation whose induced fidelity geometry remains effective across the investigated increases in interaction order, while the evaluated generic distance and polynomial kernels do not isolate the relevant sparse interactions. The advantage arises from the quantum feature-map construction and induced similarity geometry, but the present factorized kernel is also exactly classically evaluable. The results therefore demonstrate predictive and representational value that persists across the investigated scaling settings, rather than universal superiority or computational quantum advantage.

\paragraph{Data availability.}
The synthetic data are generated programmatically. The real benchmark datasets used in this study are publicly available from their respective sources and through the Fraud Dataset Benchmark.

\paragraph{Code availability.}
Code is available from the authors upon reasonable request.

\paragraph{Author contributions.}
All authors contributed to the study design, analysis, and writing of the manuscript.

\appendix

\section{Numerical kernel implementation}
\label{secA1-qk}

The mathematical feature map, closed-form block overlap, circuit resources, and experimental settings are presented in the main text. This appendix records the additional numerical details needed to connect that construction to the reported implementation, together with the derivation of the hardware estimator used in Section~\ref{subsec:hardware-protocol-qk}.

\subsection{Numerical backends}

The reference circuit is implemented with PennyLane's \texttt{default.qubit} simulator using $k$ reusable wires and the gate sequence defined in Sec.~\ref{sec:qkernel-qk}. The $d=9$ and $d=20$ synthetic settings, the real-data experiments, and the $(d,k,B)=(60,6,10)$ scaling setting generate their block states through this simulator.

For the $(d,k,B)=(80,8,10)$ experiment only, the block states are evaluated in NumPy from the exact amplitudes in Eq.~\eqref{eq:block_state_expansion-qk}. This change avoids the repeated simulator decomposition of a broadcast eight-qubit Pauli rotation. It changes only the numerical evaluation path, not the feature map or kernel definition. The NumPy statevectors were compared directly with the corresponding PennyLane circuit outputs and were required to agree to a maximum absolute error below $10^{-10}$.

\subsection{Batched Gram-matrix construction}

For a dataset $X\in\mathbb{R}^{n\times d}$, the implementation reshapes the inputs into $B=d/k$ blocks. Let $S_b(X)\in\mathbb{C}^{n\times 2^k}$ contain the block-$b$ statevectors as rows. For two datasets $X$ and $Z$, the block Gram matrix is evaluated by batched matrix multiplication,
\begin{equation}
    K_b(X,Z)
    =\left|S_b(X)S_b(Z)^{\dagger}\right|^{\circ 2},
\end{equation}
where $|\cdot|^{\circ 2}$ denotes elementwise squared modulus. The full kernel matrix is accumulated as
\begin{equation}
    K_Q(X,Z)
    =K_1(X,Z)\circ\cdots\circ K_B(X,Z),
\end{equation}
with $\circ$ denoting the Hadamard product. The diagonal of the training Gram matrix is set to one to enforce exact unit self-fidelity. No kernel centering or positive-semidefinite correction is applied.

\section{Model selection and search design}
\label{secA2-qk}

For each imbalanced scaling run, every method uses the same training--validation partition and the same validation-F1 selection procedure. Each model searches its own method-appropriate hyperparameter grid, the configuration with the highest validation F1 is selected, and that configuration is then refitted on the complete training set before a single evaluation on the independent test set. Thus, ``the same validation-F1 procedure'' refers to the data partition, selection criterion, and refitting protocol; it does not mean that different model families use identical hyperparameters.

The scaling study uses broad multi-scale grids designed for the higher-dimensional settings. Classical kernel bandwidths are adjusted for the change in ambient dimension, ensuring that the $d=60$ and $d=80$ models are evaluated at appropriate distance scales. The quantum interaction-phase range is adjusted for the rapid decrease in the characteristic magnitude of a thin-slab block product as $k$ increases. Regularization and the remaining kernel parameters are tuned jointly over ranges covering multiple scales. These choices give both model families search ranges appropriate to each problem scale. The exact numerical grids are retained in the accompanying experiment code rather than repeated here.

\section{Hardware estimator at \texorpdfstring{$\beta=0$}{beta=0}}
\label{secA3-qk}

This section derives \eqref{eq:hardware-kernel-qk} and \eqref{eq:hardware-bloch-qk}. Let $P_b=Z_{b,1}\cdots Z_{b,k}$ and $\varphi_b=\alpha\,m_b(x)$. For $\beta=0$ the block state is $\exp(i\varphi_bP_b)\,|{+}\rangle^{\otimes k}$. Since $P_b^2=I$ and $P_b|{+}\rangle^{\otimes k}=|{-}\rangle^{\otimes k}$,
\begin{equation}
|\psi^{(b)}(x^{(b)})\rangle
=\cos\varphi_b\,|{+}\rangle^{\otimes k}+i\sin\varphi_b\,|{-}\rangle^{\otimes k}.
\end{equation}
With $\varphi_b'=\alpha\,m_b(x')$, the overlap of two block states is $\cos(\varphi_b-\varphi_b')$, so
\begin{equation}
k_Q^{(b)}(x^{(b)},x'^{(b)})
=\cos^2(\varphi_b-\varphi_b')
=\tfrac12\bigl[1+\cos(2\varphi_b-2\varphi_b')\bigr].
\end{equation}
It remains to express this quantity through measurable expectation values. To this end, let $O_b=Y_{b,1}Z_{b,2}\cdots Z_{b,k}$. Using $X|{\pm}\rangle=\pm|{\pm}\rangle$, $Z|{\pm}\rangle=|{\mp}\rangle$, $Y|{+}\rangle=-i|{-}\rangle$, and $Y|{-}\rangle=i|{+}\rangle$, the two observables act on the two basis states as
\begin{equation}
X_{b,1}|{\pm}\rangle^{\otimes k}=\pm|{\pm}\rangle^{\otimes k},
\qquad
O_b|{+}\rangle^{\otimes k}=-i\,|{-}\rangle^{\otimes k},
\qquad
O_b|{-}\rangle^{\otimes k}=i\,|{+}\rangle^{\otimes k}.
\end{equation}
Taking expectations in the block state therefore gives
\begin{align}
\langle X_{b,1}\rangle&=\cos^2\varphi_b-\sin^2\varphi_b=\cos2\varphi_b,\\
\langle O_b\rangle&=-2\sin\varphi_b\cos\varphi_b=-\sin2\varphi_b,
\end{align}
which is \eqref{eq:hardware-bloch-qk}. Hence
\begin{equation}
\mathbf r_b(x)\cdot\mathbf r_b(x')
=\cos2\varphi_b\cos2\varphi_b'+\sin2\varphi_b\sin2\varphi_b'
=\cos(2\varphi_b-2\varphi_b'),
\end{equation}
and substituting into the expression for $k_Q^{(b)}$ above gives \eqref{eq:hardware-kernel-qk}.

\bibliographystyle{unsrtnat}
\bibliography{reference}

\end{document}